\pdfoutput=1 

\newif\ifconference
\conferencetrue

\ifconference
\documentclass[11pt,letter]{article}
\else
\documentclass[11pt,a4paper]{article}
\fi
\usepackage[utf8]{inputenc}
\usepackage[english]{babel}
\usepackage{microtype}
\usepackage{graphicx}
\usepackage{subfigure}
\usepackage{booktabs} 
\usepackage[hyperindex,breaklinks,bookmarks]{hyperref}
\usepackage[dvipsnames,svgnames]{xcolor}
\hypersetup{
	bookmarksdepth=2,
    unicode=true,          
    colorlinks=true,        
    linkcolor=Blue,          
    citecolor=DarkGreen,        
    filecolor=magenta,      
    urlcolor=cyan           
}

\usepackage{amsmath,amsfonts}
\usepackage{amssymb}
\usepackage[showonlyrefs]{mathtools}
\usepackage{amsthm}
\ifconference
\usepackage[portrait, left=1in, right=1in, top=1in, bottom=1in]{geometry}
\else
\usepackage[a4paper, portrait, left=1in, right=1in, top=1in, bottom=1in]{geometry}
\fi
\usepackage[numbers,sectionbib,longnamesfirst]{natbib}
\usepackage[shortlabels]{enumitem}
\usepackage[toc]{appendix}
\usepackage[capitalize,nameinlink,noabbrev]{cleveref}

\newcommand{\pp}{\operatorname{PolyPlace}}

\newtheorem{theorem}{Theorem}
\newtheorem{corollary}[theorem]{Corollary}
\newtheorem{definition}[theorem]{Definition}
\newtheorem{lemma}[theorem]{Lemma}

\newtheorem{claim}[theorem]{Claim}
\newtheorem{fact}[theorem]{Fact}

\theoremstyle{remark}

\crefname{fact}{Fact}{Facts}    

\newcommand{\eps}{\varepsilon}
\DeclareMathOperator*{\sgn}{sgn}
\def\dder#1#2{\frac{\d#1}{\d#2}}
\def\der#1{\frac{\d}{\d#1}}
\newcommand\R{\mathbb{R}}
\newcommand\datasets{\mathcal{D}}
\let\SS=\undefined
\DeclareMathOperator{\LS}{LS}
\DeclareMathOperator{\SS}{SS}
\DeclareMathOperator{\Var}{Var}
\DeclareMathOperator{\Lap}{Lap}
\def\d{d}
\let\cref\Cref

\newcommand{\eqdef}{\overset{\text{def}}{=}}

\edef\ourtitle{Smooth Sensitivity Revisited: Towards Optimality}

\date{}
\title{\ourtitle}

\author{
	Richard Hladík\thanks{Supported by the VILLUM Foundation grant 54451. The work was done while this author was visiting BARC at the University of Copenhagen.}\\
    \texttt{rihl@uralyx.cz}\\
	ETH Zurich
    \and
	Jakub Tětek\thanks{Supported by the VILLUM Foundation grant 54451.}\\
	\texttt{j.tetek@gmail.com}\\
	BARC, Univ.~of Copenhagen
}

\begin{document}
\maketitle

\begin{abstract}
Smooth sensitivity is one of the most commonly used techniques for designing practical differentially private mechanisms. In this approach, one computes the smooth sensitivity of a given query $q$ on the given input $D$ and releases $q(D)$ with noise added proportional to this smooth sensitivity. One question remains: what distribution should we pick the noise from?

In this paper, we give a new class of distributions suitable for the use with smooth sensitivity, which we name the $\pp$ distribution. This distribution improves upon the state-of-the-art Student's T distribution in terms of standard deviation by arbitrarily large factors, depending on a ``smoothness parameter'' $\gamma$, which one has to set in the smooth sensitivity framework. Moreover, our distribution is defined for a wider range of parameter $\gamma$, which can lead to significantly better performance.

Moreover, we prove that the $\pp$ distribution converges for $\gamma \rightarrow 0$ to the Laplace distribution and so does its variance. This means that the Laplace mechanism is a limit special case of the $\pp$ mechanism. This implies that out mechanism is in a certain sense optimal for $\gamma \rightarrow 0$.
\end{abstract}

\ifconference
\thispagestyle{empty}
\newpage
\clearpage
\setcounter{page}{1}
\fi

\section{Introduction}
Differential privacy has in recent years become the golden standard for data privacy. Possibly the most common way of achieving differential privacy for numerical queries is the Laplace mechanism, which adds noise from the Laplace distribution scaled with the global sensitivity of the query.\footnote{The global sensitivity is defined as the greatest difference in the query value between some two neighboring datasets.} However, global sensitivity has one significant shortcoming: it is not adaptive to the dataset at hand. Namely, the error on every input is the same as that on the most difficult input, which determines the global sensitivity.

One of the most common ways to get practical mechanisms with instance-adaptive performance is to use smooth sensitivity. In contrast with global sensitivity, the smooth sensitivity can adapt to a specific dataset, meaning that easier inputs can have significantly lower errors. The approach works by computing the $\gamma$-smooth sensitivity (\Cref{def:smooth_sensitivity}) of the given query $q$ on the input $D$ at hand for a smoothness parameter $\gamma$, and releasing $q(D)$ with added noise scaled with the $\gamma$-smooth sensitivity, picked from an appropriately chosen distribution. The parameter $\gamma$ here controls a tradeoff between how large the $\gamma$-smooth sensitivity is on one hand, and how concentrated a distribution we can use for the noise on the other. This leaves one important question unanswered: what distribution should the noise be sampled from?

In this paper, we give a new distribution class suitable for use with smooth sensitivity, which we name the $\pp$ distribution. The specific distribution from the class that we use depends on the smoothness parameter $\gamma$ and the privacy parameter $\epsilon$. Our distribution has two advantages over the state-of-the-art Student's T distribution: (1) it achieves significantly smaller error, having standard deviation smaller by arbitrarily large factors (depending on $\gamma$), as shown in \Cref{fig:plot_stddev}, and (2) it is defined for any $\gamma < \eps$, whereas the previously used distributions all required $\gamma < \eps/2$. While this may seem as an improvement by a constant factor, it may in fact bring asymptotic decrease of the error as smooth sensitivity can change a lot even with a small change in $\gamma$. Our method also significantly improves upon the Laplace distribution for all but tiny values of $\gamma$ (Laplace distribution moreover provides only approximate DP when used with smooth sensitivity).
One of the main strengths of our mechanism is that it is completely general and it provides an improvement in any application relying on smooth sensitivity.

We also prove that for $\gamma \rightarrow 0$, our $\pp$ distributions converge in distribution to the Laplace distribution and the variance of $\pp$ also converges to that of the Laplace distribution. At the same time, with $\gamma \rightarrow 0$, the $\gamma$-smooth sensitivity $\SS^\gamma$ converges to the global sensitivity. This means that in some sense, our proposed mechanism generalizes the Laplace mechanism, which is a limit special case of our mechanism for $\gamma = 0$. Since the Laplace mechanism is optimal for pure differential privacy \cite{fernandes2021laplace}, this means that our method is also optimal for $\gamma \rightarrow 0$ in the sense that any quantile of $\pp$ converges to that of the optimal distribution.
Our result can be summarized as follows:
\begin{theorem}[Informal version of \Cref{thm:main_theorem}]
Let us have a query $q$ and two parameters $0 < \gamma < \eps$. Let $D$ be a dataset. There exists a class of distributions $\pp(s, \alpha)$ such that releasing
\[
q(D) + \pp\left(\frac{\SS_q^\gamma(D)}{\gamma}, \frac{\eps}{\gamma}\right)
\]
is $\eps$-DP. Moreover, for $\gamma/\eps \rightarrow 0$, the variance is up to a factor $(1+o(1))$ the same as scaling the Laplace distribution with the smooth sensitivity. Also, 
for any value $a > 0$ it holds that $\pp(\frac{a}{\gamma}, \frac{\eps}{\gamma}) \xrightarrow{\text{d}} \Lap(a/\eps)$.
\end{theorem}

\begin{figure}
    \centering
    \includegraphics[width=.6\linewidth]{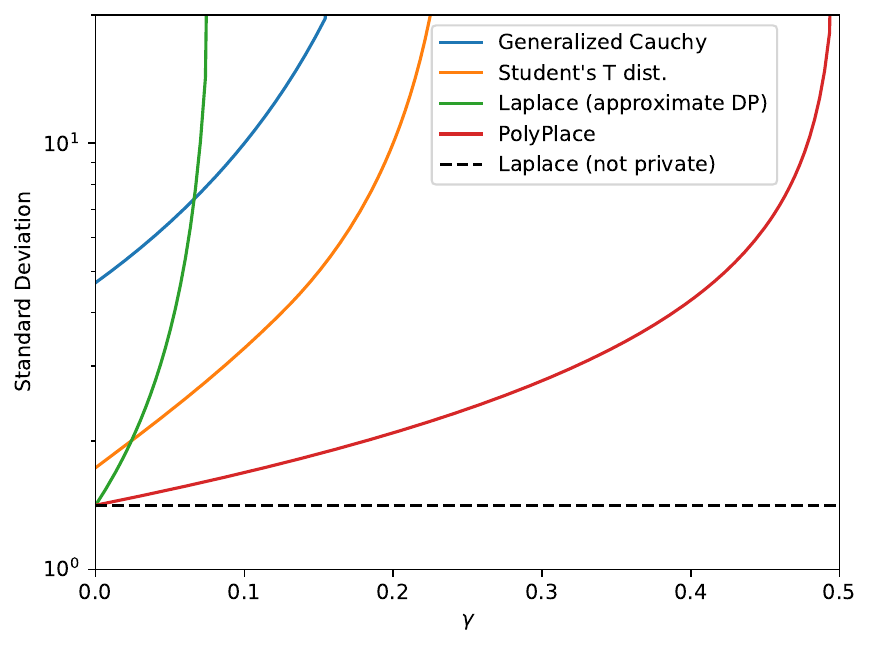}
	\caption{%
 Comparison of the standard deviation of our and the other distributions that can be used with the smooth sensitivity framework. Laplace distribution when used with smooth sensitivity only provides approximate differential privacy; we assume here $\delta = 10^{-5}$. The shape parameters for the two other distributions are numerically selected so as to minimize their standard deviations. The standard deviations are calculated in \Cref{sec:stddev_of_others}. The dashed line represents the standard deviation of the Laplace distribution which, however, needs to be used with the larger global sensitivity and not smooth sensitivity.}
    \label{fig:plot_stddev}
\end{figure}

Note that this convergence result is meaningful despite the fact that the smooth sensitivity converges to the global sensitivity for $\gamma \rightarrow 0$: This result implies that for fixed small values of $\gamma$, the distribution is close to optimal, while the smooth sensitivity will, in general, still be asymptotically smaller than the global sensitivity. This refutes the following argument, that one might naturally make against our result: Namely, that since the smooth sensitivity is, for $\gamma \rightarrow 0$, equal to the global sensitivity, this result does not show anything that would not already hold for the Laplace mechanism.

\subsubsection{Smooth Sensitivity vs. Inverse Sensitivity Mechanism}
Recently, an alternative to smooth sensitivity appeared in the form of the inverse sensitivity mechanism \cite{asi2020instance,asi2020near}. This approach has nice theoretical guarantees, and mechanisms that are optimal up to logarithmic factors have been designed by relying on inverse sensitivity. Admittedly, this reduces the overall impact of our approach. However, we are still convinced that smooth sensitivity has an important place in the repertoire of differential privacy techniques, and that it will remain being used in practice. 

The inverse sensitivity mechanism comes with its own problems. First, inverse sensitivity may be difficult to compute \cite{asi2020instance,fang2022shifted}, even in cases when smooth sensitivity is easy to compute, such as for counting triangles in a graph. While it is often possible to efficiently approximately implement the inverse sensitivity mechanism \cite{asi2020instance,fang2022shifted}, this already leads to worsened performance -- it remains to be seen whether this approach will be competitive in practice for a wide range of problems.

Moreover, there are practical considerations that often favor the smooth sensitivity approach. First, it is possible to have an efficient generic implementation just based on an oracle that returns the smooth sensitivity. This is clearly not possible for the inverse sensitivity mechanism -- the inverse sensitivity mechanism instantiates the exponential mechanism which in general cannot be implemented in time sublinear in the size of the universe. It is also due to this ease of use that there is extensive literature on the uses of smooth sensitivity on the applied side, as we discuss in \Cref{sec:related_work}. We think that even just the existence of this large amount of literature on smooth sensitivity would make improving smooth sensitivity an important undertaking.

Another issue with the inverse sensitivity mechanism is that it may in some cases have worse constants than smooth sensitivity. Specifically, it relies on the exponential mechanism which loses a factor of 2 in some instances. There are other mechanisms that could be used in place of the exponential mechanism, most notably the report-noisy-max mechanism \cite{dwork2014algorithmic}, that may have better constants, but even those lose a factor of 2 in the standard deviation in some cases. For example, the Laplace mechanism can be seen as an instantiation of the exponential mechanism or report noisy max, but in doing so, one would lose a factor of 2. Our mechanism for smooth sensitivity does not have this issue (unlike the previously known smooth sensitivity mechanisms that achieved pure differential privacy).

\subsection{Technical Overview}
The principle that motivated the $\pp$ distribution is that we wanted to make sure that for $\eps \rightarrow 0$, the privacy loss will be the same regardless of what value we released (recall that the privacy loss random variable is a function of the value that we release). Indeed, this is precisely the property that the Laplace distribution has in the context of global sensitivity. We formalize this property as a claim about some derivatives involving the density of the $\pp$ distribution (\Cref{lem:differential_equality}). Intuitively speaking, these derivatives capture an ``infinitesimal privacy loss'' as if we had two neighboring datasets whose query values as well as smooth sensitivities differ infinitesimally. In \Cref{lem:privacy_loss_bound}, we essentially integrate this infinitesimal privacy loss in order to express the true privacy loss. Doing so takes bounds on the infinitesimal privacy loss and gives bounds on the actual privacy loss, proving differential privacy.

Intuitively speaking, our improvement comes from two sources. The first is unsurprisingly that our distribution simply has a smaller variance than the Student's $T$ distribution.
The second source of our improvement is more subtle:
we do not split the privacy budget the way previous work did. Intuitively speaking, there are two sources of privacy loss when we use smooth sensitivity: from the change in the query value between adjacent inputs (corresponding to shifting of noise) and from the change in the smooth sensitivity between neighboring datasets (corresponding to scaling noise). The standard approach is to split the privacy budget between the privacy loss caused by the shifting and the privacy loss caused by scaling.
Surprisingly, if the scaling is not too large, it does not contribute anything to the worst-case privacy loss with the $\pp$ distribution. 
This allows us in some sense to use all of the privacy budget to ``pay'' for the change in query value between neighboring datasets, leading to a more efficient use of the privacy budget.


\subsection{Related Work} \label{sec:related_work}
Smooth sensitivity has been proposed by \citet{nissim2007smooth} as an alternative to global sensitivity, rectifying the global sensitivity's drawback that it is determined by the most difficult input. They have shown how it can be used to privately approximately release the median, minimum, minimum spanning tree, and the number of triangles in a graph. \citet{bun2019average} later gave new distributions suitable for use with smooth sensitivity. Notably, they propose the use of the Student's T distribution. They also propose other distributions but those only work with the weaker concentrated differential privacy \cite{dwork2016concentrated}, which we do not focus on in this paper.

Since its inception, smooth sensitivity has been used in many applications. An incomplete list includes: principal component analysis \cite{gonem2018smooth}, deep learning \cite{sun2020differentially}, statistics from genome-wide association studies \cite{yamamoto2023privacy}, generating synthetic datasets using Bayesian models \cite{li2018bayesian}, generating synthetic graphs \cite{wang2013preserving}, random forests \cite{xin2019differentially}, or some graph centrality measures \cite{laeuchli2022analysis}.

Other frameworks for beyond-worst-case differentially private mechanisms include propose-test-release \cite{dwork2009differential}, or privately bounding local sensitivity \cite[Section 3.4]{vadhan2017complexity}. However, these techniques only give approximate differential privacy and they generally speaking tend to need larger datasets to work. Namely, to get non-trivial guarantees, they usually need a dataset of size $\Omega\big(\frac{\log\delta^{-1}}{\eps}\big)$.\footnote{This is the case for propose-test-release as it otherwise always fails. For privately bounding local sensitivity, this is the case assuming the local sensitivity has global sensitivity 1 (it often has bigger sensitivity in which case the situation is even worse).} For these reasons as well as for the sake of brevity, we choose not to focus on these techniques.

\section{Preliminaries}
\subsection{Differential Privacy}

Let us have a symmetric notion of dataset adjacency $\sim$. 
A randomized algorithm $\mathcal{A}$ is said to be $\eps$-differentially private if, for any two adjacent datasets $D,D'$, and for any measurable set $S$, the following condition holds:
\begin{equation}
\Pr[\mathcal{A}(D) \in S] \leq e^{\varepsilon} \cdot \Pr[\mathcal{A}(D') \in S]
\end{equation}
where $\varepsilon$ controls the trade-off between privacy and data utility. Smaller values of $\varepsilon$ provide stronger privacy guarantees, but may limit the accuracy of the analysis. One of the most notable ways of achieving differential privacy is using the notion of smooth sensitivity:
\begin{definition}[Smooth sensitivity] \label{def:smooth_sensitivity}
Let us have a query $q: \datasets \rightarrow \R$. We define the local sensitivity of $q$ at $D \in \datasets$ as
\[
\LS_q(D) = \max \left\{\,|q(D) - q(D')|: D' \sim D\,\right\}.
\]
In turn, we define the $\gamma$-\emph{smooth sensitivity} of $q$ at $D$ as 
\[
\SS_q^{\gamma}(D)=\max \left\{\,\LS_q(D') \cdot e^{-\gamma d(D, D')}: D' \in \datasets\,\right\},
\]
	where $d(D, D')$ denotes distance induced by $\sim$. 
\end{definition}

\subsection{Calculus}
\label{sec:calculus}
We now give several statements from calculus. The following is a standard result.
\begin{fact} \label{fact:total_derivatives}
Let us have $f: \mathbb{R}^n \rightarrow \mathbb{R}$. If at a point $\mathbf{x} \in \mathbb{R}^n$, all partial derivatives exist and are continuous, then the function is differentiable at $\mathbf{x}$.
\end{fact}
For the following two statements, we choose to give statements of special cases that suffice for our needs, instead of giving the more general cases. We defer their proofs to \cref{sec:deferred_proofs}.
\begin{fact} \label{fact:splitting_derivatives}
Let us have a function $f: \R^2 \rightarrow \R$. At any $x$ at which $f(x,x)$ is differentiable, it holds
	\begin{align*}
		\frac{d}{dx}\,f(x,x)
		&=
		\left[\der{\lambda}\,f(x+\lambda,x) + \der{\lambda}\,f(x,x+\lambda)\right]_{\lambda=0} .
	\end{align*}
\end{fact}

\begin{fact} \label{fact:differentiation_scaling}
Let us have a real function $f$. It holds that
\[
\left[\frac{d}{dx}\, f(c\, x)\right]_{x = 0} = c\, \left[\frac{d}{dx}\, f(x)\right]_{x = 0} \,.
\]
\end{fact}

\section{Improved Smooth Sensitivity Mechanism}

\begin{definition} \label{def:our_distribution}
We define the $\pp(s,\alpha)$ distribution with scale $s > 0$ and shape $\alpha > 1$ by its probability density function 
\begin{align*}
	&
f_{s,\alpha}(x)
	\begin{cases}
   N_{s,\alpha }(\alpha-1)\left(1 - \frac{|x|}{s}\right)^{\alpha - 1} &\hskip-1mm \text{if } \frac{|x|}{s} < \alpha^{-1}, \\
  N_{s, \alpha} (\alpha + 1) \left(1 - \frac{1}{\alpha^2}\right)^\alpha \left(1 + \frac{|x|}{s}\right)^{-\alpha - 1} &\hskip -1mm \text{otherwise,}
\end{cases}
\end{align*}
for $N_{s,\alpha} = \frac{\alpha }{2 \left(2 \left(\frac{\alpha -1}{\alpha }\right)^{\alpha }+\alpha -1\right) s}$.
\end{definition}

\begin{figure}
    \centering
    \includegraphics[width=.6\linewidth]{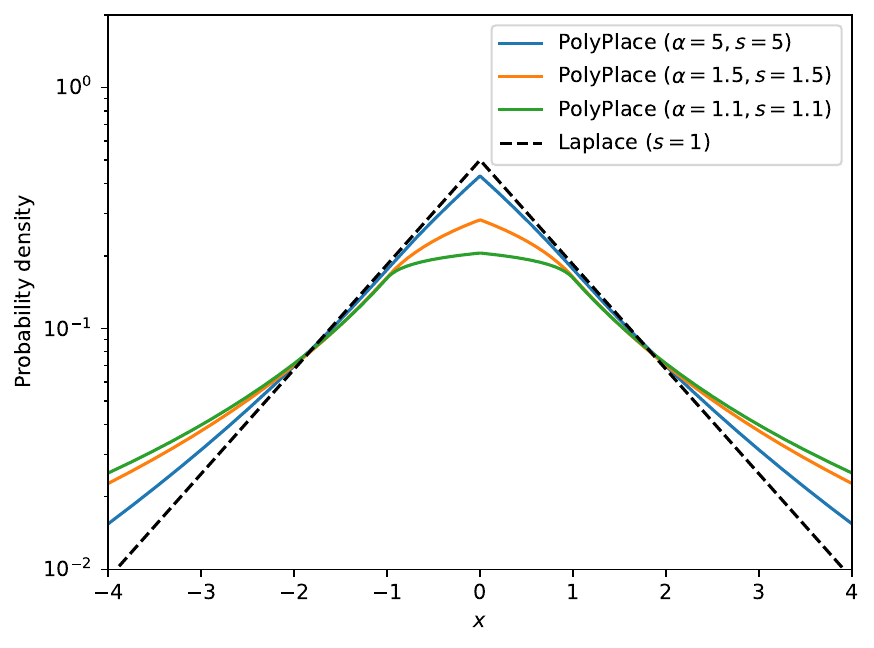}
	\caption{%
		A log-scale plot of selected members of the Laplace and $\pp$ families of distributions. The values of $\alpha$ and $s$ correspond (for $\SS(D) = 1$ and $\eps = 1$) to $\gamma \approx 0.91, \gamma = 2/3$, and $\gamma = 0.2$, respectively. Note that the distributions get closer to Laplace, which is explained by the convergence from \Cref{thm:main_theorem}.}
  \label{fig:plot_pdf}
\end{figure}

We plot a few example members of this family of distributions in \cref{fig:plot_pdf}, where we also compare it with the Laplace distribution.

We start by claiming that this distribution is correctly defined. We defer the proof to \Cref{sec:deferred_proofs}.
\begin{lemma} \label{lem:correctly_defined}
The distribution $\pp(s,\alpha)$ is correctly defined for $s >0, \alpha>1$. That is, the density is non-negative and integrates to $1$.
\end{lemma}

We now state our main result.
\begin{theorem}\label{thm:main_theorem}
Let us have a query $q: \datasets \rightarrow \R$. Let us have two parameters $0 < \gamma < \eps$. Then the mechanism
\[
M(D) = q(D) + \pp\left(\frac{\SS_q^\gamma(D)}{\gamma},\frac{\eps}{\gamma}\right)
\]
is $\eps$-DP. Moreover, for $\gamma/\eps \rightarrow 0$, it holds that 
\[
	\Var(M(D)) = (1+O(\gamma/\eps)) \, \frac{2\SS_q^\gamma(D)^2}{\eps^2} 
\]
and for any value $a > 0$ it holds that $\pp\big(\frac{a}{\gamma},\frac{\eps}{\gamma}\big) \xrightarrow{\text{d}} \Lap(a/\eps)$.
\end{theorem}
\begin{proof}
We prove the asymptotic expansion of the variance in \Cref{cor:variance_taylor}. We show that the distribution converges to Laplace in \Cref{lem:convergence_to_laplace}.

	It remains to prove privacy. Let $p_{M(D)}$ denote the density function of $M(D)$. We then have
\begin{align*}
\frac{p_{M(D')}(x)}{p_{M(D)}(x)} = \frac{f_{\SS_q^\gamma(D')/\gamma, \eps/\gamma}(x - q(D'))}{f_{\SS_q^\gamma(D)/\gamma, \eps/\gamma}(x - q(D))} \leq e^\eps \,,
\end{align*}
	where final inequality holds by \Cref{lem:privacy_loss_bound}. This implies $\eps$-differential privacy.
\end{proof}

\subsection*{Proving Privacy}
In this section, we prove \Cref{lem:privacy_loss_bound}, which is the result we used above to prove that our mechanism is differentially private. The technical parts of the proof of the lemma appear in \Cref{lem:differential_equality,lem:derivatives,lem:differentiability}. The casual reader may wish to skip the proof of \Cref{lem:privacy_loss_bound} as well as the other lemmas in this subsection.

\begin{lemma} \label{lem:privacy_loss_bound}
Let $q: \mathcal{D} \rightarrow \R$ for $\mathcal{D}$ being the set of all possible datasets, be a sensitivity $1$ query. Let us have two neighboring datasets $D \sim D'$ and $x \in \R$. Then
\[
\frac{f_{\SS_q^\gamma(D')/\gamma, \eps/\gamma}(x - q(D'))}{f_{\SS_q^\gamma(D)/\gamma, \eps/\gamma}(x - q(D))} \leq e^\eps ,
\]
where $f$ is defined in \Cref{def:our_distribution}.
\end{lemma}
\begin{proof}
	Let $x' = \frac{x-q(D)}{\SS_q^\gamma(D)}$ and define $\lambda_r$ such that $e^{\lambda_r \gamma} = \frac{\SS_q^\gamma(D')}{\SS_q^\gamma(D)}$. Let also $\lambda_s = \frac{q(D) - q(D')}{\SS_q^\gamma(D)}$ and $y(t) = e^{- \lambda_r \gamma t}(x' + \lambda_s t)$. Note that for any $s > 0, \alpha > 1$ and $c > 0$, it holds
\begin{align} \label{eq:scaling_invariance}
	f_{s, \alpha}(x) = cf_{cs, \alpha}(cx).
\end{align}
We now give a bound and we explain the individual steps below.
\allowdisplaybreaks
\begin{align}
&\frac{f_{\SS_q^\gamma(D')/\gamma, \eps/\gamma}(x - q(D'))}{f_{\SS_q^\gamma(D)/\gamma, \eps/\gamma}(x - q(D))} \nonumber\\
	&= \frac{\SS_q^\gamma(D) \cdot f_{e^{\lambda_r \gamma}/\gamma, \eps/\gamma}(\frac{x - q(D')}{\SS_q^\gamma(D)})}{\SS_q^\gamma(D) \cdot f_{1/\gamma, \eps/\gamma}(\frac{x - q(D)}{\SS_q^\gamma(D)})} \label{eq:dividing_out_ss}\\
&= \frac{f_{e^{\lambda_r \gamma}/\gamma, \eps/\gamma}(x' + \lambda_s)}{f_{1/\gamma, \eps/\gamma}(x')} \label{eq:def_of_y_prime}\\
&= \exp\left(\log f_{e^{\lambda_r \gamma}/\gamma, \eps/\gamma}(x' +  \lambda_s) - \log f_{1/\gamma, \eps/\gamma}(x') \right) \nonumber\\
	&= \exp\left(\int_0^1 \Big( \der{t} \log(f_{e^{\lambda_r t \gamma}/\gamma,\eps/\gamma}( x' + \lambda_s t))\, \Big) dt\right) \label{eq:fundamental_theorem_of_calculus}\\
	&= \exp\left(\int_0^1 \Big[ \der{\lambda} \log(f_{e^{\lambda_r (t+\lambda) \gamma}/\gamma,\eps/\gamma}( x' + \lambda_s t)) \quad+ \der{\lambda} \log(f_{e^{\lambda_r t \gamma}/\gamma,\eps/\gamma}( x' + \lambda_s (t+\lambda)))\, \Big]_{\lambda = 0} dt\right) \label{eq:split_the_derivative}\\
	&= \exp\left(\int_0^1 \Big[ \der{\lambda} \log(f_{e^{\lambda_r \lambda \gamma}/\gamma,\eps/\gamma}( y(t))) \quad + \der{\lambda} \log(f_{1/\gamma,\eps/\gamma}( y(t) + e^{-\lambda_r t \gamma} \lambda_s \lambda))\, \Big]_{\lambda = 0} dt\right) \nonumber\\
&= \exp\left(\int_0^1 \Big[ \lambda_r \der{\lambda} \log(f_{e^{\lambda \gamma}/\gamma,\eps/\gamma}( y(t))) \quad+ e^{-\lambda_r t \gamma} \lambda_s \der{\lambda} \log(f_{1/\gamma,\eps/\gamma}( y(t) +  \lambda))\, \Big]_{\lambda = 0} dt\right) \label{eq:pull_out_constants}\\
&\leq \exp\left(\int_0^1 \bigg[\lambda_r \Big|\der{\lambda} \log(f_{e^{\lambda \gamma}/\gamma,\eps/\gamma}( y(t)))\Big| \quad+  e^{-\lambda_r t \gamma} \lambda_s\Big|\der{\lambda} \log(f_{1/\gamma,\eps/\gamma}( y(t) + \lambda))\, \Big| \bigg]_{\lambda = 0} dt\right) \nonumber\\
&\leq \exp\left(\int_0^1 \bigg[\Big|\der{\lambda} \log(f_{e^{\lambda \gamma}/\gamma,\eps/\gamma}( y(t)))\Big| \quad+ \Big|\der{\lambda} \log(f_{1/\gamma,\eps/\gamma}( y(t) + \lambda))\, \Big| \bigg]_{\lambda = 0} dt\right) \label{eq:removing_constants}\\
&= \exp\left(\int_0^1 \eps \,dt\right) = e^\eps \label{eq:bounding_the_derivatives}
\end{align}
We now justify the manipulations that may not be clear.
	Equality \eqref{eq:dividing_out_ss} holds by \cref{eq:scaling_invariance}. Equality \eqref{eq:def_of_y_prime} holds by the definition of $x'$ and $\lambda_s$. The equality \eqref{eq:fundamental_theorem_of_calculus} holds because of the fundamental theorem of calculus; we use here the fact that by \Cref{lem:differentiability}, the function
 \[
 \ell(y,z) = \log(f_{e^{\lambda_r y \gamma}/\gamma,\eps/\gamma}( x' + \lambda_s z))
 \] 
	is differentiable except at points where $x' + \lambda_s z = 0$, at which points it is continuous; furthermore, the function $\phi(t) = \ell(t, t)$ is continuous everywhere and differentiable everywhere except for at most one point where $x' + \lambda_s t = 0$. The equality \eqref{eq:split_the_derivative} holds by \eqref{fact:splitting_derivatives}; we again used that $h$ is differentiable everywhere except for the at most one $t$ satisfying $x' + \lambda_s t = 0$. The equality \eqref{eq:pull_out_constants} holds by viewing each summand as a univariate function in $\lambda$ and applying \Cref{fact:differentiation_scaling}.
The bound \eqref{eq:removing_constants} holds because we argue below that the removed constants are $\lambda_r,e^{-\lambda_r t \gamma}\gamma_s \leq 1$.
The equality \eqref{eq:bounding_the_derivatives} holds by the \Cref{lem:differential_equality} below.


It remains to prove $\lambda_r,e^{-\lambda_r t \gamma}\gamma_s \leq 1$. The definition of smooth sensitivity ensures that $\frac{\SS_q^\gamma(D')}{\SS_q^\gamma(D)} \leq e^\gamma$, and thus it holds that $\lambda_r \leq 1$. It also holds
\[
q(D) - q(D') \leq \min(\SS_q^\gamma(D),\SS_q^\gamma(D'))
\]
and thus
\[
\lambda_s \leq \min\left(1,\;\frac{\SS_q^\gamma(D)}{\SS_q^\gamma(D')}\right).
\]
It thus holds $e^{-\lambda_r t \gamma}\lambda_s \leq 1$ as we needed to prove.
\end{proof}

In the following lemma, we explicitly calculate the derivatives of $\log f_{s,\alpha}(x)$ with respect to $s$ and $x$, which is useful for \cref{lem:differentiability,lem:differential_equality}.

\begin{lemma} \label{lem:derivatives}
	For $\alpha > 1$ and $h : \R_{>0} \times \R \to \R$ defined as $h(s, x) =
	\log f_{s,\alpha}(x)$, it holds for any $x \ne 0$:
	\begin{align*}
		\left[\der{\lambda} h(s, x + \lambda)\right]_{\lambda=0} &= \sgn(x)\cdot
\begin{dcases}
	\frac{\alpha - 1}{|x|-s}
	& \text{if } \frac{|x|}{s} < \alpha^{-1}, \\
	\frac{-\alpha - 1}{|x|+s}
	& \text{otherwise},
\end{dcases}
\\
		\left[\der{\lambda} h(e^{\lambda/s} s, x)\right]_{\lambda=0} &=
\begin{dcases}
	-\frac1s + \frac{(\alpha - 1)|x|}{s(s - |x|)}
	& \text{if } \frac{|x|}{s} < \alpha^{-1}, \\
	-\frac1s + \frac{(\alpha+1)|x|}{s(s + |x|)}
	& \text{otherwise}.
\end{dcases}
	\end{align*}
\end{lemma}
\begin{proof}
	The statement follows from direct computation. Full proof is deferred to \cref{sec:deferred_proofs}.
\end{proof}
We proceed with a technical lemma that justifies the invocation of the fundamental theorem of calculus in the proof of \cref{lem:privacy_loss_bound}.

\begin{lemma} \label{lem:differentiability}
Let $0 < \gamma < \eps$; $\lambda_r, \lambda_s, x' \in \R$.
The function
\[
\ell(y,z) = \log(f_{e^{\lambda_r y \gamma}/\gamma,\eps/\gamma}( x' + \lambda_s z))
\]
	is continuous everywhere. It is differentiable everywhere, except when $x' + \lambda_s z = 0$.
\end{lemma}
\begin{proof}
	Set $\alpha = \eps / \gamma$ and take $h$ from \cref{lem:derivatives}. Define $y' = e^{\lambda_r y\gamma}/\gamma$ and $z' = x' + \lambda_s z$. Then $\ell(y, z) = h(y', z')$ and we have
	\begin{align*}
		\der{y} \ell(y, z) = \left[\ell(y + \lambda, z)\right]_{\lambda=0} &= \left[h(y' \cdot e^{\lambda_r \lambda \gamma}, z')\right]_{\lambda=0},
		\\
		\der{z} \ell(y, z) = \left[\ell(y, z + \lambda)\right]_{\lambda=0} &= \left[h(y', z' + \lambda_s\lambda)\right]_{\lambda=0}.
	\end{align*}
	Now we can apply \cref{lem:derivatives} on the right-hand sides, together
	with the chain rule, to obtain that the partial derivatives on the left exist and
	are continuous for all $(y, z)$ such that $z' = x' + \lambda_s z \ne 0$. By
	\cref{fact:total_derivatives} and the fact that around every $(y, z)$,
	there is a sufficiently small neighborhood where $x' + \lambda_s z \ne 0$,
	we get that $\ell$ is differentiable (and thus continuous) in every $(y,
	z)$ with $x' + \lambda_s z \ne 0$.

	It remains to prove that $\ell$ is also continuous for $x' + \lambda_s z =
	0$. It suffices to show that $h(s, x)$ is continuous at $(s, 0)$ for
	all choices of $s > 0$. The first branch $h_1$ of $h$ is continuous
	everywhere, and for any $(s, 0)$, there exists a sufficiently small
	$\eps$-neighborhood of $(s, 0)$ where $|x| / s < \alpha^{-1}$, that is,
	where $h(s, x) = h_1(s, x)$. Therefore, $h$ is continuous at $(s, 0)$.
\end{proof}

\begin{lemma} \label{lem:differential_equality}
	Let $x \in \R \setminus \{0\}$ and let $0< \gamma < \eps$. It holds
	\begin{align}
		\label{eq:differential_equality}
		&\Bigg[\bigg|\der{\lambda} \log(f_{1/\gamma, \eps/\gamma}(x + \lambda))\Bigg|
		+ \Bigg|\der{\lambda} \log(f_{e^{\lambda \gamma}/\gamma, \eps/\gamma}(x)))\bigg|  \Bigg]_{\lambda = 0}  = \eps.
	\end{align}
\end{lemma}
\begin{proof}[Proof sketch]
	For the full proof, see \cref{sec:deferred_proofs}. The main idea of the proof is to calculate both expressions in the sum using \cref{lem:derivatives}, with $s = 1/\gamma$ and $\alpha = \eps/\gamma$. Then, we distinguish two cases: either $|x| / s < \alpha^{-1}$, or $|x| / s \ge \alpha^{-1}$. For both cases separately, and for each absolute value, we can argue that the expression inside this absolute value will always have the same sign, and can be therefore written without the absolute value. This simplifies the analysis, and we can then verify that the left-hand side of \cref{eq:differential_equality} always sums up to $\alpha / s$. Finally, we substitute $s$ and $\alpha$ back to get $\alpha / s = \eps$.
\end{proof}

\subsection{Properties of the $\pp(s, \alpha)$ Distribution}
We now give a lemma that among others states the variance of $\pp$. To make sure we have not made a mistake in the very technical manipulations, we have checked that we get the same result if we approximate the variance by numerical integration.\footnote{Specifically, by scaling, we can assume w.l.o.g.\ that $s = 1$. For this value of $s$, we have plotted the variance as calculated by numeric integration and the expression above. We checked that they are indeed the same.} We defer the proof to \cref{sec:deferred_proofs}.
\begin{lemma} \label{lem:variance}
$\pp$ has mean $0$ and variance
{\small
\begin{align*}
\sigma^2_{s,\alpha} = 2 s^2 \frac{\left(\frac{1}{\alpha }+1\right)^{-\alpha } \left(\left(19 \alpha ^2+5\right) \left(1-\frac{1}{\alpha ^2}\right)^{\alpha }+(\alpha -2) (\alpha -1)^2 \left(\frac{1}{\alpha }+1\right)^{\alpha }\right)}{\left(2 \left(\frac{\alpha -1}{\alpha }\right)^{\alpha }+\alpha -1\right) \left(\alpha ^4-5 \alpha ^2+4\right)}
	.
\end{align*}
}
\end{lemma}

\begin{corollary}  \label{cor:variance_taylor}
Assume $\gamma < \eps/2$. It holds
\[
	\Var(M(D)) = (1+O(\gamma/\eps)) \, \frac{2 \SS_q^\gamma(D)^2}{\eps^2} .
\]
\end{corollary}
\begin{proof}
	We substitute into the above lemma with $s = \SS_q^\gamma(D)/\gamma$, $\alpha = \eps/\gamma$, and take the Taylor expansion of order one w.r.t.\ $\gamma$ around $\gamma = 0$. We skip routine steps in the computation and just state that the second-order Taylor expansion is
\[
	\frac{2 {\SS_q^\gamma(D)}^2}{\eps^2}-\frac{2 (3 e-17) \gamma {\SS_q^\gamma(D)}^2}{e \eps^3}+O\left(\gamma^2\right).
\]
This is exactly as we claim.
\end{proof}

Note that the variance of the Laplace distribution scaled to the smooth sensitivity is $\frac{2}{\eps^2}$, meaning that for $\gamma/\eps \rightarrow 0$, the variance of our mechanism approaches that of the Laplace mechanism scaled to the smooth sensitivity.
In fact, the distribution converges to the Laplace distribution for $\gamma \rightarrow 0$.
\begin{lemma} \label{lem:convergence_to_laplace}
For $\gamma \rightarrow 0$ and any $a,\eps > 0$, it holds that $f_{a/\gamma, \eps/\gamma}$ converges pointwise to the probability density function (pdf) of $\Lap(a/\eps)$.
\end{lemma}
\begin{proof}
Let us denote the pdf of $\Lap(a/\eps)$ by $f_{\Lap(a/\eps)}$.
We prove pointwise convergence, that is, that for any $x \in \R$, $f_{a/\gamma, \eps/\gamma}(x) \rightarrow f_{\Lap(a/\eps)}(x)$. This in turn implies the desired convergence in distribution.\footnote{This holds because by Scheffé's lemma, this means that the TV distance of the two distributions goes to zero, which in turn implies convergence in distribution.} Since $f_{\Lap(a/\eps)}$ is bounded, this follows if we prove that for $\gamma/\eps \rightarrow 0$ and any fixed $a>0$ and $x \in \R$, it holds that $f_{a/\gamma, \eps/\gamma}(x) / f_{\Lap(a/\eps)}(x) \rightarrow 1$, denoted $f_{a/\gamma, \eps/\gamma}(y) \sim f_{\Lap(a/\eps)}(y)$. We prove this claim in the rest of this proof. 

Set $s = a / \gamma$, $\alpha = \eps / \gamma$. It holds that
\begin{align*}
	\bar{f}_1 &\eqdef \left(1-\frac{|x|}{s}\right)^{\alpha-1} = \left(1-\frac{|x| \gamma}{a}\right)^{\eps/\gamma-1} \sim \left(1-\frac{|x| \gamma}{a}\right)^{\eps/\gamma} =  \left(\left(1-\frac{|x| \gamma}{a}\right)^{a/(\gamma |x|)}\right)^{\eps |x|/a} \sim e^{- \eps |x| / a},
\end{align*}
where the first $\sim$ uses the fact that the base of the exponential converges to $1$ and the second $\sim$ uses the standard limit $(1-z)^{1/z} \rightarrow e^{-1}$.
At the same time, it holds that
\begin{align*}
	c \eqdef \frac{\alpha }{2 \left(2 \left(\frac{\alpha -1}{\alpha }\right)^{\alpha }+\alpha -1\right) s} &= \frac{\eps/\gamma }{2 \left(2 \left(\frac{\eps/\gamma -1}{\eps/\gamma }\right)^{\eps/\gamma }+\eps/\gamma -1\right) a/\gamma} \\&\sim \frac{\eps/\gamma}{2 \left(2 e^{-1}+\eps/\gamma -1\right) a/\gamma} \\&\sim \frac{\eps/\gamma}{2  \cdot \eps/\gamma \cdot a/\gamma} = \frac{\gamma}{2 a},
\end{align*}
where the first $\sim$ again uses the limit $(1-z)^{1/z} \rightarrow e^{-1}$. Overall, we can thus for $\frac{|x|}{s} < \alpha^{-1}$ write
\begin{align*}
	f_{a/\gamma, \eps/\gamma}(x) &= c \Big(\frac{\eps}{\gamma} - 1\Big) \bar{f}_1 \sim \frac{\gamma}{2 a} \Big(\frac{\eps}{\gamma} - 1\Big) e^{- \eps |x| / a} \sim \frac{\gamma}{2 a} \cdot \frac{\eps}{\gamma} e^{- \eps |x| / a} = \frac{\eps}{2 a} e^{- \eps |x| / a} = f_{\Lap(a/\eps)}(x),
\end{align*}
as we wanted to show.
The argument for $\frac{|x|}{s} \geq \alpha^{-1}$ is very similar:
	\begin{align*}
		\bar{f}_2 &\eqdef \left(1+\frac{|x|}{s}\right)^{-\alpha-1} = \left(1+\frac{|x| \gamma}{a}\right)^{-\eps/\gamma-1}
		\\
		&\sim \left(1+\frac{|x| \gamma}{a}\right)^{-\eps/\gamma}
	\left(\left(1+\frac{|x| \gamma}{a}\right)^{-a/(\gamma |x|)}\right)^{\eps |x|/a}
		\sim e^{- \eps |x| / a}
	\end{align*}
and
\begin{align*}
	f_{a/\gamma, \eps/\gamma}(x)
	&=
	c \, \left(\frac{\eps}{\gamma} + 1\right) \bigg(1-\frac{\gamma^2}{\eps^2}\bigg)^{\eps/\gamma} \bar{f}_2
	\\
	&\sim \frac{\gamma}{2 a} \left(\frac{\eps}{\gamma} + 1\right) \bigg(1-\frac{\gamma^2}{\eps^2}\bigg)^{\eps/\gamma}  e^{- \eps |x| / a}
	\\
	&\sim \frac{\gamma}{2 a} \frac{\eps}{\gamma} e^{- \eps |x| / a}
	\\
	&= \frac{\eps}{2 a} e^{- \eps |x| / a} = f_{\Lap(a/\eps)}(x),
\end{align*}
where the second $\sim$ uses the standard limit $(1-x^2)^{1/x} \rightarrow 1$ for $x \rightarrow 0$.
\end{proof}

\section*{Acknowledgements}
We would like to thank Rasmus Pagh and Adam Sealfon for helpful discussions. We would like to thank Rasmus Pagh for hosting Richard Hladík at the University of Copenhagen.

\bibliographystyle{plainnat}
\bibliography{literature}

\appendix
\onecolumn
\section{Deferred Proofs} \label{sec:deferred_proofs}
We start by proving the calculus facts from \cref{sec:calculus}.

\begin{proof}[Proof of \cref{fact:splitting_derivatives}]
	The statement follows from the chain rule. Its special case states that for $f : \R^2 \to \R$ and differentiable $g : \R \to \R$ and $h : \R \to \R$, and all $x \in \R$ such that $y = g(x)$, $z=h(x)$ and $f$ is differentiable at $(y, z)$, it holds
	\begin{align*}
		&\dder{\,f(y, z)}{x}
		=
		\dder{\,f(y, z)}{y}
		 \cdot 
		\dder yx
		+
		\dder{\,f(y, z)}{z}
		 \cdot
		\dder zx
		=
		\left[
		\dder{\,f(y + \lambda, z)}{\lambda}
		 \cdot 
		\dder yx
		+
		\dder{\,f(y, z + \lambda)}{\lambda}
		 \cdot
		\dder zx
		\right]_{\lambda=0},
	\end{align*}
	where the notation on right-hand side is an equivalent way of writing
	partial derivatives. In our case, $g = h = \text{id}$, $x = y = z$ and $\dder yx
	= \dder zx = 1$, and substituting these facts in the equation above
	concludes the proof.
\end{proof}
\begin{proof}[Proof of \cref{fact:differentiation_scaling}]
	The statement follows from the chain rule. Its special case states that if $h(x) = f(g(x))$, $g$ is
	differentiable everywhere and $f$ is differentiable at $g(x)$, then $h'(x)
	= f'(g(x)) \cdot g'(x)$. If we apply the chain rule on $g(x) = cx$, then $h(x) =
	f(cx)$ and we get that $\frac{\d}{\d x} f(cx) = h'(x) = f'(g(x)) \cdot
	g'(x) = f'(cx) \cdot c$. Evaluated at $x = 0$, this becomes $c \cdot
	f'(0)$, as needed.
\end{proof}

Next, we prove \Cref{lem:correctly_defined}, which states that the $\pp$ distribution is correctly defined.
\begin{proof}[Proof of \Cref{lem:correctly_defined}.]
By scaling, we may assume without loss of generality that $s = 1$. By symmetry, it then suffices to integrate $f_{1,\alpha}$ from $0$ to $+\infty$ and verify that it integrates to $1/2$. The first branch in the definition of $f_{1,\alpha}$ integrates to\footnote{We do not show the detailed steps of how to compute the antiderivative. However, it should be noted that it is easy to check correctness by differentiating and checking that this results in the respective expressions we were integrating.}
\[
I_1(\alpha,x) = -\frac{(\alpha -1) (1-x)^{\alpha }}{2 \left(2 \left(\frac{\alpha -1}{\alpha }\right)^{\alpha }+\alpha -1\right)},
\]
and the second integrates to
\[
I_2(\alpha,x) = -\frac{\left(1-\frac{1}{\alpha ^2}\right)^{\alpha } (\alpha +1) (x+1)^{-\alpha }}{2 \left(2 \left(\frac{\alpha -1}{\alpha }\right)^{\alpha }+\alpha -1\right)}.
\]
We now write the integral as
\[
\int_{0}^\infty f_{1,\alpha}(x) \, dx =  I_1(\alpha, 1/\alpha) - I_1(\alpha, 0) + \lim_{x \rightarrow +\infty} I_2(\alpha, x) - I_2(\alpha, 1/\alpha) = 1/2 ,
\]
where in the last step we skipped some elementary manipulations.
\end{proof}

Next, we prove \cref{lem:derivatives}, which provides a formula for the derivatives of $h(s, x) = \log f_{s,\alpha}(x)$.

\begin{proof}[Proof of \cref{lem:derivatives}]
	Define 
	\begin{equation*}
		g(s, x, m) = N_{s,\alpha}(\alpha-m)\left(1 - \frac{mx}{s}\right)^{m\alpha - 1}.
	\end{equation*}
	Then one can verify that
	\begin{equation*}
		h(s, x) = \begin{cases}
	\log g(s, |x|, +1) & \text{if } \frac{|x|}{s} < \alpha^{-1}, \\
	\log g(s, |x|, -1) + \alpha\log \left(1 - \frac{1}{\alpha^2}\right) & \text{otherwise},
\end{cases}
	\end{equation*}
	When taking derivatives, we can ignore the additive $\alpha\log\left(1 - \frac{1}{\alpha^2}\right)$ in the second branch as $\alpha$ is constant.

	We start with analyzing the derivatives of $g$ with respect to $x$. For $m = \pm 1, x \ne 0$:
	\begin{align*}
		\left[\der\lambda\log(g(s, |x + \lambda|, m))\right]_{\lambda=0}
	&=
		\left[\underbrace{\der\lambda\log\left(N_{s,\alpha} \cdot (\alpha-m)\right)}_{0} + \der\lambda (m\alpha - 1) \log \left(1 - \frac{m|x + \lambda|}{s}\right)\right]_{\lambda=0}
	\\&=
		\left[(m\alpha - 1) \frac{\sgn(x + \lambda)m}{m|x + \lambda|-s}\right]_{\lambda=0}
		= \sgn(x) \cdot
\begin{dcases}
	\frac{\alpha - 1}{|x|-s} & \text{if } m = 1, \\
	\frac{-\alpha - 1}{|x|+s} & \text{otherwise}.
\end{dcases}
	\end{align*}
%
	Now, for $x \ne 0$, $|x| / s \ne \alpha^{-1}$, it holds:
	\begin{align}
		\label{eq:partial_wrt_x}
		\left[\der{\lambda} h(s, x + \lambda)\right]_{\lambda=0}
		&=
		\sgn(x) \cdot
\begin{dcases}
	\frac{\alpha - 1}{|x|-s}
	& \text{if } \frac{|x|}{s} < \alpha^{-1}, \\
	\frac{-\alpha - 1}{|x|+s}
	& \text{otherwise}.
\end{dcases}
	\end{align}

	Finally, for $|x| / s = \alpha^{-1}$, each branch gives us a one-sided derivative at $x = \pm s / \alpha$, and one can verify that both derivatives are equal and therefore \cref{eq:partial_wrt_x} holds for all $x \ne 0$.

	Similarly, for $m = \pm1, x \ne 0$:
	\begin{align*}
		\left[\der\lambda\log(g(e^{\lambda/s}s, |x|, m))\right]_{\lambda=0}
	&=
		\bigg[\der\lambda\log N_{e^{\lambda/s}s,\alpha} + \underbrace{\der\lambda\log (\alpha-m)}_0 + \der\lambda(m\alpha - 1) \log \left(1 - \frac{m|x|}{e^{\lambda/s}s}\right)\bigg]_{\lambda=0}
	\\&=
		\left[\der\lambda\log e^{-\lambda/s} + \der\lambda(m\alpha - 1) \log \left(1 - \frac{m|x|}{e^{\lambda/s}s}\right)\right]_{\lambda=0}
	\\&=
		\left[-\frac1s + \frac{(m\alpha - 1)\cdot m|x|}{s\left(e^{\lambda/s}s - m|x|\right)}\right]_{\lambda=0}
	=
\begin{dcases}
	-\frac1s + \frac{(\alpha - 1)|x|}{s(s - |x|)}
		& \text{if } m = 1, \\
	-\frac1s + \frac{(\alpha+1)|x|}{s(s + |x|)}
		& \text{otherwise}.
\end{dcases}
	\end{align*}
	Thus, for $x \ne 0, |x| / s \ne \alpha^{-1}$:
	\begin{align}
		\label{eq:partial_wrt_exp_s}
		\left[\der{\lambda} h(e^{\lambda/s}s, x)\right]_{\lambda=0}
	=
\begin{dcases}
	-\frac1s + \frac{(\alpha - 1)|x|}{s(s - |x|)}
		& \text{if } \frac{|x|}{s} < \alpha^{-1}, \\
	-\frac1s + \frac{(\alpha+1)|x|}{s(s + |x|)}
		& \text{otherwise}.
\end{dcases}
	\end{align}
	Again, the two branches are equal for $|x| / s = \alpha^{-1}$ and \cref{eq:partial_wrt_exp_s} hence holds for all $x \ne 0$.
\end{proof}

\begin{proof}[Proof of \cref{lem:differential_equality}]
	In this whole proof, fix $s > 0, \alpha > 1$. We will carry through the proof with $s$ and $\alpha$, and in the very end substitute $s = 1/\gamma$ and $\alpha = \eps/\gamma$.

	Per \cref{lem:derivatives}, it holds:
	\begin{align}
		\Bigg[\left|\der\lambda\log(f_{s, \alpha}(x + \lambda))\right|\Bigg]_{\lambda=0}
		&=
\begin{dcases}
	\left|\frac{\alpha - 1}{|x|-s}\right|
	=
	\frac{\alpha - 1}{s - |x|}
	& \text{if } \frac{|x|}{s} < \alpha^{-1}, \\
	\left|\frac{-\alpha - 1}{|x|+s}\right|
	=
	\frac{\alpha + 1}{s + |x|}
	& \text{otherwise}.
\end{dcases}
	\end{align}
	To justify the removal of the absolute value on the right-hand side, note that $\alpha > 1$ and thus $\alpha - 1 > 0$ and $-\alpha - 1 < 0$. Then consider two cases: if $|x| / s < \alpha^{-1} < 1$, then $|x| - s < 0$ and thus $(\alpha - 1) / (|x| - s) < 0$, and if $|x| / s \ge \alpha^{-1} > 0$, then $|x| + s > 0$ and thus $(-\alpha - 1) / (|x| + s) < 0$.

	Similarly,
	\begin{align*}
		\Bigg[\left|\der\lambda\log(f_{e^{\lambda/s}s, \alpha}(x))\right|\Bigg]_{\lambda=0}
	&=
\begin{dcases}
	\left|
	-\frac1s + \frac{(\alpha - 1)|x|}{s(s - |x|)}
	\right|
	=
	\frac1s - \frac{(\alpha - 1)|x|}{s(s - |x|)}
		& \text{if } \frac{|x|}{s} < \alpha^{-1}, \\
	\left|
	-\frac1s + \frac{(\alpha+1)|x|}{s(s + |x|)}
	\right|
	=
	-\frac1s + \frac{(\alpha+1)|x|}{s(s + |x|)}
		& \text{otherwise}.
\end{dcases}
	\end{align*}
	To again justify the removal of absolute values, we distinguish two cases. If $|x| / s < \alpha^{-1}$, then
	\begin{align*}
		\frac{s}{|x|} > \alpha
		\implies
		\frac{s - |x|}{|x|} > \alpha - 1
		\implies
		\frac{s - |x|}{(\alpha - 1)|x|} > 1
		\implies
		\frac{(\alpha - 1)|x|}{s - |x|} < 1
		\implies
		\frac{(\alpha - 1)|x|}{s(s - |x|)} < \frac1s,
	\end{align*}
	and thus, $\frac1s - \frac{(\alpha - 1)|x|}{s(s - |x|)} > 0$. If $|x| / s \ge \alpha^{-1}$, then, similarly,
	\begin{align*}
		\frac{s}{|x|} \le \alpha
		\implies
		\frac{s + |x|}{|x|} \le \alpha + 1
		\implies
		\frac{s + |x|}{(\alpha + 1)|x|} \le 1
		\implies
		\frac{(\alpha + 1)|x|}{s + |x|} \ge 1
		\implies
		\frac{(\alpha + 1)|x|}{s(s + |x|)} \ge \frac1s,
	\end{align*}
	and thus, $\frac1s - \frac{(\alpha + 1)|x|}{s(s + |x|)} \le 0$.

	Now we can combine both results to conclude that
	\begin{align*}
		&\Bigg[\left|\der\lambda\log(f_{s, \alpha}(x + \lambda))\right|+ 
		\left|\der\lambda\log(f_{e^{\lambda/s}s, \alpha}(x))\right|\Bigg]_{\lambda=0}
		\\
		&=
\left.\begin{dcases}
	\frac{\alpha - 1}{s - |x|}
	+
	\frac1s - \frac{(\alpha - 1)|x|}{s(s - |x|)}
	=
	\frac{(s - |x|) (1 + \alpha - 1)}{s(s - |x|)}
		&\text{if } \frac{|x|}{s} < \alpha^{-1},\\
	\frac{\alpha + 1}{s + |x|}
	-
	\frac1s + \frac{(\alpha+1)|x|}{s(s + |x|)}
	=
	\frac{(s + |x|) (\alpha + 1 - 1)}{s(s + |x|)}
		&\text{otherwise};
\end{dcases}\right\} = \frac{\alpha}{s}.
	\end{align*}

	To conclude the proof, set $s = 1 / \gamma$, $\alpha = \eps/\gamma$ and observe that $s > 0$ and $\alpha > 1$ as needed. Thus, all previous claims hold, and we have
	\begin{align*}
		\Bigg[\left|\der\lambda\log(f_{1/\gamma, \eps/\gamma}(x + \lambda))\right| + 
		\left|\der\lambda\log(f_{e^{\lambda\gamma}/\gamma, \eps/\gamma}(x))\right|\Bigg]_{\lambda=0}
		&= \frac{\eps/\gamma}{1/\gamma} = \varepsilon.
	\end{align*}
\end{proof}

Now we prove \cref{lem:variance}, which states that the $\pp$ distribution has mean 0 and calculates its variance.

\begin{proof}[Proof of \cref{lem:variance}.]
Assuming the expectation is defined, it has to be $0$ as the distribution is symmetric around this point. We now focus on the variance; by proving that the variance is finite, we also get that the expectation is defined and therefore $0$. Similarly to proof of \Cref{lem:correctly_defined}, we first separately give the indefinite integrals of the two branches, this time of $x^2 f_{1,\alpha}(x)$:
\[
I_1^{\sigma^2} = -\frac{(\alpha -1) (1-x)^{\alpha } \left(\alpha  (\alpha +1) x^2+2 \alpha  x+2\right)}{2 (\alpha +1) (\alpha +2) \left(2 \left(\frac{\alpha -1}{\alpha }\right)^{\alpha }+\alpha -1\right)}
\]
\[
I_2^{\sigma^2} = \frac{\left(1-\frac{1}{\alpha ^2}\right)^{\alpha } (\alpha +1) (x+1)^{-\alpha } \left(-\left((\alpha -1) \alpha  x^2\right)-2 \alpha  x-2\right)}{2 (\alpha -2) (\alpha -1) \left(2 \left(\frac{\alpha -1}{\alpha }\right)^{\alpha }+\alpha -1\right)}.
\]
We may now use this to write the integral as
\begin{align*}
\sigma_{1,\alpha}^2 &= 2\int_0^{+\infty} x^2 f_{1,\alpha}(x) \, dx = I_{1}^{\sigma^2}\Big(\alpha, \frac{1}{\alpha}\Big) - I_{1}^{\sigma^2}(\alpha, 0) + \lim_{{x \to +\infty}} I_{2}^{\sigma^2}(\alpha, x) - I_{2}^{\sigma^2}\Big(\alpha, \frac{1}{\alpha}\Big) \\&= \frac{\left(\frac{1}{\alpha }+1\right)^{-\alpha } \left(\left(19 \alpha ^2+5\right) \left(1-\frac{1}{\alpha ^2}\right)^{\alpha }+(\alpha -2) (\alpha -1)^2 \left(\frac{1}{\alpha }+1\right)^{\alpha }\right)}{\left(2 \left(\frac{\alpha -1}{\alpha }\right)^{\alpha }+\alpha -1\right) \left(\alpha ^4-5 \alpha ^2+4\right)} \,.
\end{align*}
\end{proof}

\section{Other Mechanisms for Smooth Sensitivity} \label{sec:stddev_of_others}
In this section, we calculate the variances of the generalized Cauchy distribution of \cite{nissim2007smooth}, the Student's~T distribution \cite{bun2019average}, and the Laplace distribution, when used with the smooth sensitivity framework.
%
%
We use these variances in \Cref{fig:plot_stddev} in order to compare our distribution with these two distributions.



\begin{claim}
When scaled according to \citet{nissim2007smooth}, the generalized Cauchy distribution has, for $c > 3$ and $0 < \gamma < \eps/(c+1)$, a standard deviation
\[
\frac{c+1}{(\eps-\gamma (c+1)) \sqrt{2 \cos \left(\frac{2 \pi }{c}\right)+1}} \,.
\]
For other values of $c>0,\gamma>0$, the standard deviation is infinite.
\end{claim}
\begin{proof}
The generalized Cauchy distribution is defined in \citet{nissim2007smooth} as having density $\propto \frac{1}{1+|y|^\alpha}$.
Using a Computer Algebra System (such as Mathematica), it is easy to verify that the density is in fact equal to
\[
\frac{c}{2 \pi  \csc \left(\frac{\pi }{c}\right) \left(| y| ^c+1\right)}.
\]
Writing the variance as an integral and again using a CAS to evaluate it, one gets that the variance is
\[
\frac{1}{2 \cos \left(\frac{2 \pi }{c}\right)+1}
\]
for $c > 3$ and infinite otherwise.
For two neighboring datasets $D \sim D'$, let us have $\lambda_r$ such that $SS_q^\gamma(D) = e^{\lambda_r} SS_q^\gamma(D')$ and let $\lambda_s = \frac{q(D') - q(D)}{SS_q^\gamma(D)}$.
It is shown in \citet{nissim2007smooth} that the privacy loss is at most
$(|\lambda_r| + |\lambda_s|) (c+1)$.
If we substitute the bound $|\lambda_r| \leq \gamma$, the remaining privacy budget is $\eps - \gamma (c+1)$, meaning that we need to scale the distribution with
$
\frac{c+1}{\eps - \gamma (c+1)}
$ in order to ensure the bound on privacy loss is at most $\eps$, assuming $\gamma < \eps/(c+1)$. This results in a standard deviation of 
\[
\frac{c+1}{(\eps-\gamma (c+1)) \sqrt{2 \cos \left(\frac{2 \pi }{c}\right)+1}}.
\]
\end{proof}

\begin{claim}
When scaled according to \citet{bun2019average}, the Student's T distribution has, for $d > 2$ and $0 < \gamma < \eps/(d+1)$, a standard deviation
\[
\sqrt{\frac{d}{d-2}}\cdot\frac{d+1}{2\sqrt{d}(\eps - \gamma(d+1))}
\]
For other values of $d>0,\gamma>0$, the standard deviation is infinite.
\end{claim}

\begin{proof}
The variance of the T distribution with $d$ degrees of freedom is $\frac{d}{d-2}$. At the same time, for $\lambda_r$, $\lambda_s$ defined as in the proof above, the privacy loss is at most
\[
|\lambda_r| \cdot(d+1)+|\lambda_s| \cdot \frac{d+1}{2 \sqrt{d}}.
\]
We have $|\lambda_r| \leq \gamma$. The remaining budget for privacy loss caused by the distribution shift is $\eps - \gamma\,(d+1)$, meaning that we need
\[
|\lambda_s| \leq \nu = \frac{2\sqrt{d}(\eps - \gamma(d+1))}{d+1}
\]
assuming $\gamma < \eps/(d+1)$.
To achieve $\eps$-differential privacy, for a fixed value of $\gamma$, we have to scale the distribution with $\nu^{-1}$. This results in a standard deviation of
\[
\sqrt{\frac{d}{d-2}}\cdot\frac{d+1}{2\sqrt{d}(\eps - \gamma(d+1))}.
\]
\end{proof}

\begin{claim}
When scaled according to \citet{nissim2007smooth}, the Laplace distribution has, for $0 < \gamma < \eps/\log \delta^{-1}$, a standard deviation
\[
\frac{\sqrt{2}}{\eps - \gamma \log 2\delta^{-1}}.
\]
\end{claim}
\begin{proof}
The proof is essentially the same as the proofs above. It is shown by \cite{nissim2007smooth} that the privacy loss is at most $|\lambda_r| \log 2\delta^{-1} + |\lambda_s|$. This means we have to scale the noise to be inversely proportional to $\eps - \gamma \log 2\delta^{-1}$. This leads to the claimed standard deviation.
\end{proof}

\end{document}